\documentclass{article}

\usepackage{amssymb}
\usepackage{hyperref}
\usepackage{amsmath}
\usepackage{graphicx}%
\usepackage{multirow}%
\usepackage{amsmath,amssymb,amsfonts}%
\usepackage{amsthm}%
\usepackage{mathrsfs}%
\usepackage[title]{appendix}%
\usepackage{xcolor}%
\usepackage{textcomp}%
\usepackage{manyfoot}%
\usepackage{booktabs}%
\usepackage{algorithm}%
\usepackage{algorithmicx}%
\usepackage{algpseudocode}%
\usepackage{listings}%
\usepackage{dsfont}
\usepackage{braket}
\usepackage{amssymb}
\usepackage{ragged2e}
\usepackage{subfig}
\usepackage{float}
\usepackage{hyperref}
\usepackage[utf8]{inputenc} 
\usepackage[T1]{fontenc}    
\usepackage{hyperref}       
\usepackage{url}            
\usepackage{booktabs}       
\usepackage{amsfonts}       
\usepackage{nicefrac}       
\usepackage{microtype}      
\usepackage{lipsum}
\usepackage{fancyhdr}       
\usepackage{graphicx} 
\usepackage{indentfirst}
\newtheorem{theorem}{Theorem}
\newtheorem{proposition}[theorem]{Proposition}%
\newtheorem{corollary}[theorem]{Corollary}

\theoremstyle{thmstyletwo}%

\theoremstyle{thmstylethree}%
\title{Tight entropy bound based on p-quasinorms.}

\author{
  Juan Pablo Lopez \\
  Universidad del Valle \\
  Cali, Colombia\\
  juan.lopez.holguin@correounivalle.edu.co \\
}

\begin{document}
\maketitle

\begin{abstract}
\noindent In the present paper we prove a family of tight upper and lower bounds for the Shannon entropy and von Neumann entropy based on the p-norms. This allows us to have an entropy estimate, a criterion for the finiteness of it and a bound on the difference of entropy, additionally, we did some numerical tests that show the efficiency of our approximations.
\end{abstract}

\noindent\textbf{Keywords: } Entropy,  Estimate, Bound.

\section{Introduction}

\indent The Shannon entropy is one of the most important information quantities and it have been widely studied since it's introduction by Claude Shannon \cite{Shannon}. This is a measure defined on probability distributions that can have infinite number of values, then if we have a probability distribution $\textbf{p}=(p)_i$ the Shannon entropy is defined as 
\begin{equation}\label{Shannon entropy}
    S(\textbf{p})=-\sum_i p_i\text{log}_2 (p_i), 
\end{equation}
where $0\text{log}_2 (0)=0$. The basis of the logarithm will be always 2 and we write simply as log.\\
\indent There is a special interest in finding bounds for (\ref{Shannon entropy}) with other information quantities or metrics in the space of probability distributions \cite{Extremal relations, Relations,Linder, Tapus,Simic,Yamin,Yamin 2,Bacetti}. In the Section \ref{Sec Log ine} we prove and optimize a logarithm inequality that will allow us to have an upper and lower bounds for the entropy in the Section \ref{Sec entro bound}. Our bounds are based on the p-quasinorms defined on $\mathbb{R}^n$ as
\begin{equation}\label{P norm}
    ||\textbf{x}||_p=\left( \sum_i^n x^p \right)^{\frac{1}{p}},
\end{equation}
with $0<p<\infty$ and if $\textbf{x}\in\ell^{p}(\mathbb{R})$ as 
\begin{equation}\label{P inf norm}
    ||\textbf{x}||_p=\left(\sum_i^\infty x^p \right)^{\frac{1}{p}}.
\end{equation}
\indent With $ 1\leq p$ we have that (\ref{P norm}) and (\ref{P inf norm}) define norms \cite{Lax analysis}. Our upper bounds depends on the p-quasinorms with $0<p<1$ and works in infinite dimensions, with this we came to the same sufficient condition as in \cite{Bacetti} for (\ref{Shannon entropy}) to be finite in the Theorem \ref{Theo entro bound}.\\
\indent The results presented in the Section \ref{Sec entro bound} are valid for the von Neumann entropy $H(\cdot)$ defined on the quantum density matrices. \\
\indent Finally, our bounds can be used to approximate the entropy and the difference of entropy between two probability distributions, this is tested in the Section \ref{Sec Numerical} with a numerous amount of probability distributions.

\section{Logarithm inequality}\label{Sec Log ine}

\indent In this section we present a family of inequalities that allow us to bound the entropy, additionally we find the tight forms of these inequalities.

\begin{proposition}
Let $x\in [0,1]$, then there exists some positive constants $C_1,C_2$ such that
\begin{equation}\label{log inequality}
    C_1 \left(1-x^{1-\sigma} \right) \leq -\text{log}\left(x\right)\leq C_2 \left(\frac{1}{x^{1-\sigma}}-1 \right),
\end{equation}
for every $\sigma\in(0,1)$. Moreover, the best choice of constants is 
\[C_1=C_2=\frac{1}{\text{ln}(2)(1-\sigma)}.\]
\end{proposition}
\begin{proof}
\indent Consider the quotients 
\begin{equation}\label{l(x)}
    l(x)=\frac{-\text{log}\left(x\right)}{1-x^{1-\sigma}},
\end{equation}
\begin{equation}\label{S(x)}
    s(x)=\frac{-\text{log}\left(x\right)}{x^{\sigma-1}-1},
\end{equation}
for $x\in (0,1]$. We will prove that $l(x)$ and $s(x)$ are decreasing and increasing, respectively, this allows us to set a minimum and maximum at x=1.\\
\indent Note that these are differentiable functions on $(0,1]$ and we only need to prove that $l'(x)\leq 0$ and $s'(x)\geq 0$, respectively. \\
\indent The inequality $l'(x)\leq 0$ is equivalent to 
\begin{equation}\label{l(x) alternative}
x^{1-\sigma}(1-\text{ln}(x)(1-\sigma))\leq 1. 
\end{equation}
\indent Consider the function $\Tilde{l}(x)=x^{1-\sigma}(1-\text{ln}(x)(1-\sigma))$, this function is increasing on $(0,1]$ because we easily see that $\Tilde{l}'(x)\geq 0 $. Now we conclude that the maximum of $\Tilde{l}(x)$ is at $x=1$, which is precisely 1, proving (\ref{l(x) alternative}).\\
\indent On the other hand, the inequality $s'(x)\geq 0$ is equivalent to 
\begin{equation}\label{s(x) alternative}
x^{\sigma-1}(1+\text{ln}(x)(1-\sigma))\leq 1. 
\end{equation}
\indent Let $\Tilde{s}(x)=x^{\sigma-1}(1+\text{ln}(x)(1-\sigma))$, this function is increasing on $(0,1]$ because $\Tilde{s}'(x)\geq 0 $, then the maximum of $\Tilde{s}(x)$ is at $x=1$, which is precisely 1, proving (\ref{s(x) alternative}).\\
\indent Finally, the fact that $l(x)$ and $s(x)$ are decreasing and increasing, respectively, allows us to find the best constants as 
\begin{equation}\label{l(x) min}
    C_1=\lim_{x\to 1}\frac{-\text{log}\left(x\right)}{1-x^{1-\sigma}}=\frac{1}{\text{ln}(2)(1-\sigma)},
\end{equation}
\begin{equation}\label{S(x) max}
    C_2=\lim_{x\to 1}\frac{-\text{log}\left(x\right)}{x^{\sigma-1}-1}=\frac{1}{\text{ln}(2)(1-\sigma)},
\end{equation}

\end{proof}

Now we have an infinite amount of upper and lower bounds for the entropy, but which is the best? This question will be answered in the following proposition.

\begin{proposition}\label{l and s are good}
    Let $\sigma,\Tilde{\sigma}\in (0,1)$, if $\sigma<\Tilde{\sigma}$ then
    \begin{equation}
        \frac{1}{1-\sigma}\left(\frac{1}{x^{1-\sigma}}-1 \right) \geq \frac{1}{1-\Tilde{\sigma}}\left(\frac{1}{x^{1-\Tilde{\sigma}}}-1 \right),
    \end{equation}
    \begin{equation}
        \frac{1}{1-\sigma}\left(1-x^{1-\sigma} \right) \leq \frac{1}{1-\Tilde{\sigma}}\left(1-x^{1-\Tilde{\sigma}} \right) , 
    \end{equation}
    for all $x\in [0,1]$. Additionally, we have that
    \begin{equation}
        \lim_{\sigma\to 1^-} \frac{1}{1-\sigma}\left(\frac{1}{x^{1-\sigma}}-1 \right) = -\text{ln}(x),
    \end{equation}
    \begin{equation}
        \lim_{\sigma\to 1^-} \frac{1}{1-\sigma}\left(1-x^{1-\sigma} \right) = -\text{ln}(x) , 
    \end{equation}    
\end{proposition}
\begin{proof}
    Let $f(\sigma)=\frac{1}{1-\sigma}\left(\frac{1}{x^{1-\sigma}}-1 \right)$, $f(\sigma)$ is differentiable on [0,1], then we only need to verify that $f'(\sigma)\leq 0$. The inequality $f'(\sigma)\leq 0$ becomes 
    \[ \frac{1}{(1-\sigma)^2}\left(\frac{1}{x^{1-\sigma}}-1 \right) +\frac{\text{ln}(x)x^{\sigma-1}}{1-\sigma} \leq 0, \]
    which is equivalent to left side of (\ref{log inequality}).\\
    \indent Let $g(\sigma)=\frac{1}{1-\sigma}\left(1-x^{1-\sigma} \right)$,  the inequality $g'(\sigma)\geq 0$ becomes 
    \[\frac{1-x^{1-\sigma}}{(1-\sigma)^2}+\frac{\text{ln}(x)x^{1-\sigma}}{1-\sigma}\geq 0,\]
    which is equivalent to right side of (\ref{log inequality}). The limits are simply an application of the L'Hôpital's rule \cite{Lax calculus}.
\end{proof}

The Proposition \ref{l and s are good} guaranties that  if $\sigma\to 1^{-}$ we obtain a better inequalities. In the next section we use the inequality (\ref{log inequality}) to bound and estimate the Shannon entropy.

\section{Entropy bound}\label{Sec entro bound}
Using (\ref{log inequality}) we can formulate the inequality 
\begin{equation}\label{xlogx inequality}
    \frac{1}{ln(2)(1-\sigma)} \left(x-x^{2-\sigma} \right) \leq -x\text{log}\left(x\right)\leq \frac{1}{ln(2)(1-\sigma)} \left(x^{\sigma}-x \right),
\end{equation}
for all $x\in [0,1]$ and every $\sigma\in (0,1)$. We are ready to obtain our bound for the entropy in the following proposition.
\begin{theorem}\label{Theo entro bound}
   Let $\textbf{p}=(p_1,...,p_n)$ be a discrete probability distribution and $\sigma\in (0,1)$, then 
   \begin{equation}\label{finite inequality}
        \frac{1}{ln(2)(1-\sigma)}(1-||\textbf{p}||_{2-\sigma})\leq S(\textbf{p}) \leq  \frac{1}{ln(2)(1-\sigma)} (||\textbf{p}||_\sigma -1 ).
   \end{equation}
   \indent If $\textbf{p}$ is an infinite discrete distribution, we have that $\textbf{p}\in \ell^{\sigma}(\mathbb{R})$ implies 
    \begin{equation}\label{inf inequality}
        \frac{1}{ln(2)(1-\sigma)}(1-||\textbf{p}||_{2-\sigma})\leq S(\textbf{p}) \leq  \frac{1}{ln(2)(1-\sigma)} (||\textbf{p}||_\sigma -1 ),
   \end{equation} 
    this inequality provides the same criterion for the finiteness of the entropy presented in \cite{Bacetti}. Moreover, the inequalities (\ref{finite inequality}) and (\ref{inf inequality}) are tight for every $\sigma\in (0,1)$.
\end{theorem}
\begin{proof}
Using (\ref{xlogx inequality}) in the sum of the terms $p_i\text{log}(p_i)$ and the fact that $\sum_i p_i=1$ to deduce our inequalities. In the case of the infinite discrete distribution we need to ensure that $\sum_i p_i^{\sigma}<\infty$, that is why we add the condition $\textbf{p}\in \ell^{\sigma}(\mathbb{R})$.\\
\indent If we set $\textbf{p}=(1,0,0,..)$ we have that $||\textbf{p}||_{2-\sigma}=||\textbf{p}||_{\sigma})=1$ and $S(\textbf{p})=0$, then our inequalities become an equalities for every .
\end{proof}
\indent We need to guarantee that our estimates are good enough, here below we will prove that our bounds get closer as $\sigma \to 1^-$.
\begin{corollary}
    Let $\textbf{p}$ be a discrete probability distribution such that $\textbf{p}\in \ell^{\sigma}(\mathbb{R})$ for some $\sigma\in (0,1)$, then 
\begin{equation}\label{ine error above}
 \left| S(\textbf{p}) -  \frac{1}{ln(2)(1-\sigma)} (||\textbf{p}||_\sigma -1 ) \right|\leq  \frac{1}{ln(2)(1-\sigma)} (||\textbf{p}||_\sigma +||\textbf{p}||_{2-\sigma}  -2 ),  
\end{equation}
\begin{equation}\label{ine error down}
 \left| S(\textbf{p}) -  \frac{1}{ln(2)(1-\sigma)} (1-||\textbf{p}||_{2-\sigma}) \right|\leq  \frac{1}{ln(2)(1-\sigma)} (||\textbf{p}||_\sigma +||\textbf{p}||_{2-\sigma}  -2 ) .  
\end{equation}
\end{corollary}

\begin{corollary}
    Let $\textbf{p},\textbf{q}$ be a discrete probability distributions such that $\textbf{p},\textbf{q}\in \ell^{\sigma}(\mathbb{R})$ for some $\sigma\in (0,1)$ and $S(\textbf{p})\geq  S(\textbf{q})$, then 
\begin{equation}\label{diff ine above}
  S(\textbf{p}) - S(\textbf{q}) \leq \frac{1}{ln(2)(1-\sigma)} (||\textbf{p}||_\sigma +||\textbf{q}||_{2-\sigma}  -2 ).  
\end{equation}
\indent Moreover, there exists a $\sigma\in (0,1)$ such that
\begin{equation}\label{diff ine below}
  S(\textbf{p}) - S(\textbf{q}) \geq \frac{1}{ln(2)(1-\sigma)} (2-||\textbf{p}||_{2-\sigma} -||\textbf{q}||_{\sigma}).  
\end{equation}
\end{corollary}
\indent The inequality (\ref{diff ine below}) can not be formulated for all $\sigma\in (0,1)$ because is possible that the upper bound of $S(\textbf{q})$ is greater than $S(\textbf{p})$. But as our bounds are close enough to the entropy with $\sigma\to 1^-$ we can ensure that for some $\sigma$ the upper bound of $S(\textbf{q})$ is greater than $S(\textbf{p})$ is smaller than $S(\textbf{p})$.\\
\indent In the following section we test the capacity of our bounds to estimate the entropy and the entropy difference of probability distributions.

\section{Numerical experiments}\label{Sec Numerical}
We are going to test the capacity of our bounds to estimate the entropy and the difference of entropies, to do this, we generate randomly 500 probability distributions with 100 values each one.\\
\indent We set $\sigma=0.9$ for all tests, we need to have in mind the limitations of the machines in the calculus of $x^\sigma$ and $x\text{log}(x)$. All the programming codes used in this paper are publicly available in \cite{J lopez}.\\
\indent The following figure shows the values of the estimates obtained with (\ref{finite inequality}).
\begin{figure}[H]
\centering
\includegraphics[width=0.52\textwidth]{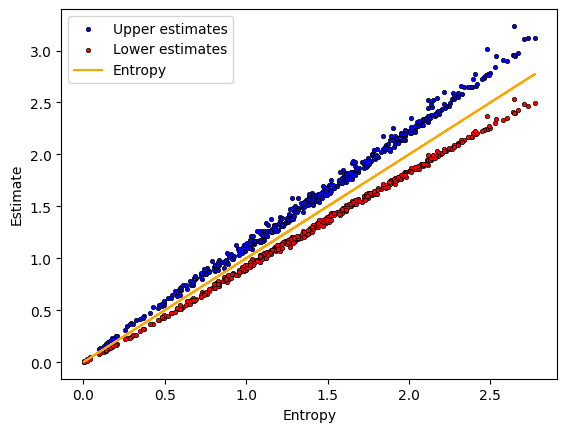}
\caption{Upper and Lower entropy bounds.}
\label{Fig Entropy Upp Low}
\end{figure}
\indent The orange line is the identity line and represents the value of the entropy, we can see that the upper bound is above and the lower bound is below this line.\\
\indent To analyze with more precision the effectiveness, we calculate the absolute  and relative error. In the Figure \ref{Fig ent bound errors} we see that with larger entropy we have a larger absolute error, on the other hand, we have that with entropies close to zero we could have a larger relative error but in the larger entropies this is more stable. Despite this, in general the bounds behave good as it is show in the Figure \ref{Fig Entropy Upp Low}.

\begin{figure}[H]\label{Fig ent bound errors}
    \centering
    \subfloat[\centering Entropy bounds absolute error. ]{{\includegraphics[width=4.9cm]{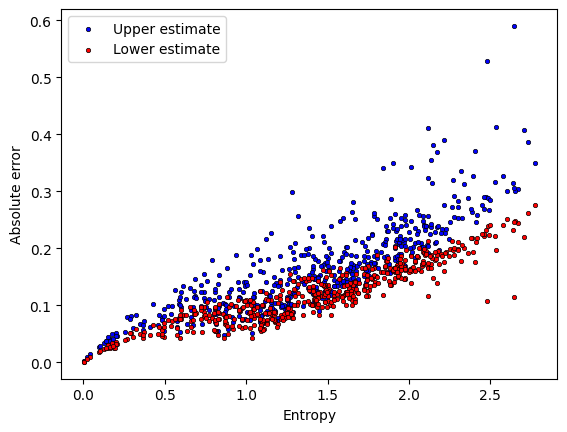} }}%
    \qquad
    \subfloat[\centering Entropy bounds absolute error.]{{\includegraphics[width=4.9cm]{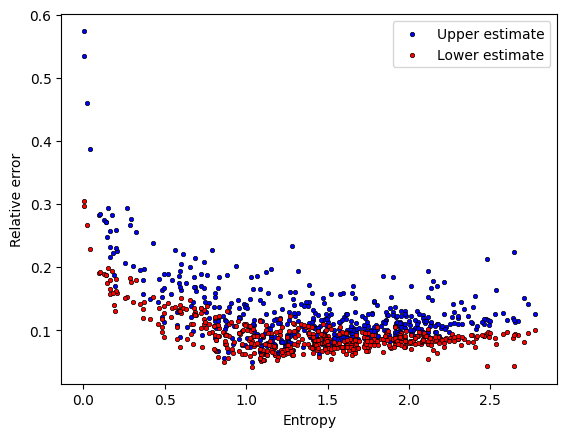} }}%
    \caption{Entropy bounds absolute and relative errors}%
    \label{Fig ent bound errors}%
\end{figure}
\indent To test the bound (\ref{ine error above}) we generate randomly 500 pairs of probability distributions, each pair we calculate the absolute difference and compare it with  with our upper bound. Again we use the orange line to represent the real value of the entropy difference and all.
\begin{figure}[H]
\centering
\includegraphics[width=0.55\textwidth]{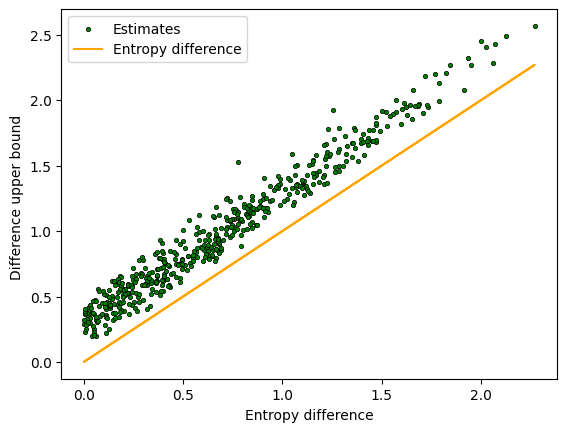}
\caption{Entropy difference bound .}
\end{figure}
\indent We can observe that our bound is very close to the real values of the differences, but we need to take a close look when the difference is zero.\\
\indent The Figure \ref{Fig diff entro error} shows the absolute and relative error of our bound, the absolute error is small but seems not to depend on the entropy difference. On the other side, the relative error shows that our bound behaves good except when the entropy difference is near to zero, but this is compensated with a low absolute error.
\begin{figure}[H]
    \centering
    \subfloat[\centering Entropy difference absolute error. ]{{\includegraphics[width=5.3cm]{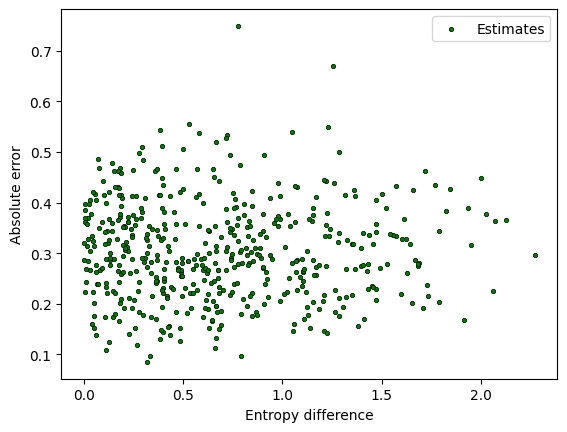} }}%
    \qquad
    \subfloat[\centering Entropy difference relative error]{{\includegraphics[width=5.3cm]{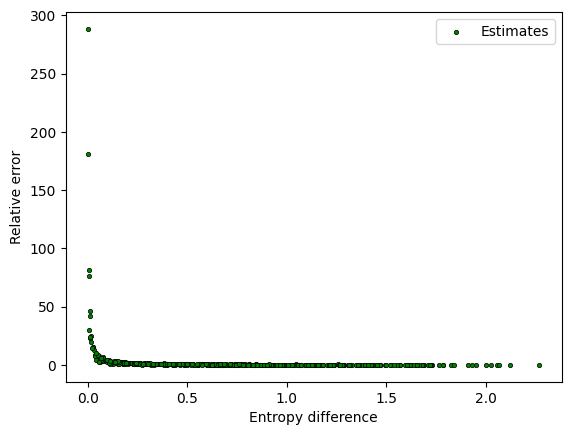} }}%
    \caption{Entropy difference bound absolute and relative errors.}%
    \label{Fig diff entro error}%
\end{figure}

\section{Conclusions}

We obtained a new but simple bounds for entropy, that can be used in infinite dimensions and for the von Neumann entropy. These bounds give an explicit relation between the entropy and the p-norms, in addition these bounds could be close enough to be used as approximations. Our finiteness criterion for the entropy coincides with one given in \cite{Bacetti}. Additionally, we conjecture that our bounds on the entropy difference could be improved to some form of uniform continuity bounds.

\bibliographystyle{unsrt}

\end{document}